\documentclass[leqno, 11pt]{article}
\usepackage{times}
\usepackage{amssymb,amscd}
\usepackage{amsmath}
\usepackage{amsfonts}
\usepackage{mathrsfs}
\usepackage{xfrac}
\usepackage{comment}

\usepackage{amsthm}

\newtheorem{theorem}{Theorem}[section]
\newtheorem{definition}[theorem]{Definition}
\newtheorem{remark}[theorem]{Remark}

\newtheorem{lemma}[theorem]{Lemma}
\newtheorem{corollary}[theorem]{Corollary}

\newtheorem{fact}[theorem]{Fact}

\newcommand{\secret}[1]{}

\begin{document}

\title{Belief, knowledge and evidence}

\author{Steffen Lewitzka\thanks{Universidade Federal da Bahia -- UFBA,
Instituto de Computação,
Departamento de Ciência da Computação,
40170-110 Salvador BA,
Brazil,
steffenlewitzka@web.de} 
\text{ }and Vinícius Pinto\thanks{Universidade Federal da Bahia -- UFBA,
Instituto de Computação,
Departamento de Ciência da Computação,
40170-110 Salvador BA,
Brazil}}

\maketitle

\begin{abstract}
We present a logical system that combines the well-known classical epistemic concepts of belief and knowledge with a concept of evidence such that the intuitive principle \textit{`evidence yields belief and knowledge'} is satisfied. Our approach relies on previous works of the first author \cite{lewjlc2, lewigpl, lewapal} who introduced a modal system containing $S5$-style principles for the reasoning about intutionistic truth (i.e. \textit{proof}) and, inspired by \cite{artpro}, combined that system with concepts of \textit{intuitionistic} belief and knowledge. We consider that combined system and replace the constructive concept of \textit{proof} with a classical notion of \textit{evidence}. This results in a logic that combines modal system $S5$ with classical epistemic principles where $\square\varphi$ reads as `$\varphi$ is evident' in an epistemic sense. Inspired by \cite{lewapal}, and in contrast to the usual possible worlds semantics found in the literature, we propose here a relational, frame-based semantics where belief and knowledge are not modeled via accessibility relations but directly as sets of propositions (sets of sets of worlds).      
\end{abstract}

\section{Introduction}

The goal of this work is to combine the well-known epistemic concepts of belief and knowledge, as axiomatized in classical modal epistemic logic (see, e.g. \cite{fag, mey}), with an appropriate concept of \textit{evidence} (or \textit{certainty}), to formalize that combination by a single logical system and to present an intuitive semantics. 

An introduction of \textit{evidence} (or \textit{certainty}) into the context of classical epistemic logic requires some conceptual clarifications. Which are the basic intuitions and principles behind our notion of evidence? How these intuitions can be adequately formalized (axiomatized) in combination with the usual epistemic concepts of belief and knowledge? 

Intuitively, evidence can be regarded as a kind (a strong form) of reason or justification for knowledge and belief. There are also approaches found in the philosophical literature where evidence is identified with knowledge itself.\footnote{We do not aim for a deeper, qualified philosophical discussion. Our intention here is to present a very basic clarifaction and justification of our conceptual approach.} In any case, the implication 
\begin{equation*}
    \text{evidence }\Longrightarrow \text{ knowledge, belief and truth}
\end{equation*}
is surely acceptable. Moreover, we shall assume here that evidence differs from knowledge. Of course, in a constructive, intuitionistic setting, evidence or certainty can be understood as effected constructions, as \textit{proof} in the intuitive sense of BHK (Brouwer-Heyting-Kolmogorov) interpretation, i.e. as intuitionistic truth itself. If we interpret evidence in this constructive way as intuitionistic truth, then Intuitionistic Epistemic Logic (IEL) introduced by Artemov and Protopopescu \cite{artpro} can be seen as a combination of evidence (\textit{proof}) with belief and knowledge. In IEL, the epistemic concepts are formalized in accordance with BHK interpretation of intuitionistic logic. Knowledge of proposition $\varphi$ is given as a \textit{verification} of $\varphi$ (in an intuitive sense). Since \textit{proof} is regarded as the strongest kind of \textit{verification}, \textit{proof} yields knowledge and belief. This fundamental principle of IEL is formalized by the axiom of co-reflection $\varphi\rightarrow K\varphi$.

While \textit{proof} might be a convincing interpretation of the concept of \textit{evidence} (or \textit{certainty}) in an intuitionistic epistemic context, the question arises if some appropriate classical counterpart for classical epistemic logic can be found.\footnote{Of course, the concept of classical truth is not a candidate: the classical principle of reflection $K\varphi\rightarrow\varphi$, `knowledge yields classical truth' would imply that knowledge and evidence are equivalent concepts.}

In recent works \cite{lewjlc2, lewapal}, the first author of this article developed a hierarchy of Lewis-style modal logics for the reasoning about intuitionistic truth, i.e. \textit{proof}. In \cite{lewapal} it is argued that $L5$, the strongest logic of that hierarchy, is the most adequate system for the reasoning about proof since its $S5$-style axioms are intuitionistically acceptable in the sense of an extended BHK interpretation; moreover, $L5$ has a relational semantics based on intuitionistic frames which reflects intuitionistic reasoning in the usual and natural way. An important idea underlying that approach is the distinction between \textit{actual proofs} (effected constructions which are immediately accessible) and \textit{possible proofs} (conditions on constructions which must not be in conflict with effected constructions, i.e. actual proofs). Classically, formula $\square\varphi$ expresses that $\varphi$ is intuitionistically true, i.e. $\varphi$ has an actual proof. However, we also need an intuitionistic interpretation of $\square\varphi$ (consider, e.g. formula $\square\square\varphi$ claiming that $\square\varphi$ holds intuitionistically). The following clause for $\square\varphi$ extends the standard BHK interpretation of intuitionistic logic: ``a proof of $\square\varphi$ is given by a presentation of an actual proof of $\varphi$."\\

The axioms of modal system $L5$ are given by the following schemes:

\begin{itemize}
    \item[] (INT)  All theorems of \textit{intuitionistic} propositional logic IPC
    \item[] (A1)  $\square(\varphi\vee\psi)\rightarrow (\square\varphi\vee\square\psi)$ (Disjunction Property)
    \item[] (A2) $\square\varphi\rightarrow\varphi$ (Intuitionistic truth yields classical truth)
    \item[] (A3) $\square(\varphi\rightarrow\psi)\rightarrow (\square\varphi\rightarrow\square\psi)$ (distribution)
    \item[] (A4)  $\square\varphi\rightarrow\square\square\varphi$ 
    \item[] (A5)  $\neg\square\varphi\rightarrow\square\neg\square\varphi$
    \item[] (TND) $\varphi\vee\neg\varphi$ (Tertium non Datur)
\end{itemize}

The rules of inference are modus ponens MP and intuitionistic Axiom Necessitation AN: ``If $\varphi$ is an intuitionistically acceptable axiom (i.e. an instance of (INT) or of (A1)--(A5)), then infer $\square\varphi$ as a theorem." \\

Actually, the $S5$-style modal axioms of $L5$ are intuitionistically acceptable, i.e. justifiable on the base of the extended BHK interpretation. For instance, a proof of $\square\varphi\rightarrow\varphi$ is a construction that converts any proof of $\square\varphi$ into a proof of $\varphi$. Of course, by the extended BHK interpretation, such a construction is immediately accessible, i.e. that formula has always an actual proof and is intuitionistically valid. The presentation of an actual proof should be regarded itself as an actual proof, and an actual proof can always be presented (as it is immediately accessible). This yields a constructive justification of axiom (A4) $\square\varphi\rightarrow\square\square\varphi$.\\

A justification of principle (A5) in the sense of extended BHK interpretation is a little more complicated. However, it is important for the understanding of our interpretation of \textit{evidence} as a classical substitute for \textit{proof}, so let us examine it. First, we note that the following principle (1) is valid:
\begin{equation}\label{10}
\text{``If $\varphi$ has no actual proof, then a proof of $\square\varphi$ is impossible."}
\end{equation}

In fact, a possible proof (i.e. conditions on the construction) of $\square\varphi$ would in particular involve the necessary condition that $\varphi$ has an actual proof. By hypothesis, such an actual proof does not exist. So the above principle holds. It can be formalized in classical logic as 
\begin{equation}\label{15}
\neg\square\varphi\rightarrow\neg\Diamond\square\varphi,
\end{equation}
where $\Diamond$ is the modal operator of possibility defined by $\Diamond\varphi :=\neg\square\neg\varphi$. In the context of logic $L5$, $\Diamond\varphi$ reads ``$\varphi$ has a possible proof". The above formula then is equivalent to $\neg\square\varphi\rightarrow\square\neg\square\varphi$, i.e. to axiom (A5). This is a justification of modal principle (A5) under the classical interpretation. Below, we show that (A5) is also intuitionistically acceptable.\\

(1) implies the following
\begin{equation}\label{20}
\text{``Either $\square\varphi$ has an actual proof or $\neg\square\varphi$ has an actual proof."}
\end{equation}
In fact, if $\square\varphi$ has no actual proof, then $\varphi$ has no actual proof and principle (1) yields the impossibility of a proof of $\square\varphi$. But if $\square\varphi$ has no possible proof, then the identity function, as an immediately given construction, constitutes an actual proof of $\neg\square\varphi = \square\varphi\rightarrow\bot$.\footnote{We assume here that the connective for negation is not primitive but defined in terms of implication and falsum by $\neg\psi := (\psi\rightarrow \bot)$.} Thus, (3) holds.\\

Now let us suppose $s$ is a proof of $\neg\square\varphi$, i.e. of $\square\varphi\rightarrow\bot$. If $s$ is only a possible proof, then it involves the condition that $\square\varphi$ has no proof. As a possible proof, $s$ must not be in conflict with actual proofs. Then, by principle (3), $\neg\square\varphi = \square\varphi\rightarrow\bot$ has an actual proof. That actual proof can be examplified by the identity function. Its presentation constitutes an actual proof of $\square\neg\square\varphi$. We have presented a construction that converts any proof of $\neg\square\varphi$ into a proof of $\square\neg\square\varphi$. Thus, axiom (A5) is intuitionistically acceptable in the sense of extended BHK interpretation.\\
 
It is important to notice that, although its main subject is intuitionistic proof, $L5$ is a classical logic: classical truth of $\square\varphi$ means that $\varphi$ is intuitionistically true. Lewitzka \cite{lewigpl} combines $L5$ with adapted concepts of intuitionistic belief and knowledge which are inspired by Intuitionistic Epistemic Logic \cite{artpro}. Some of the resulting principles (in the extended language of epistemic logic) are the following:
\begin{itemize}
    \item[] (A6) $K(\varphi\rightarrow\psi)\rightarrow (K\varphi\rightarrow K\psi)$
    \item[] (A7) $\square\varphi\rightarrow \square K\varphi$ (this axiom can be equivalently replaced by $\square\varphi\rightarrow K\varphi$)
    \item[] (A8) $K\varphi \rightarrow \neg\neg\varphi$
\end{itemize}

These principles are expected to be intuitionistically valid. For instance, (A8) expresses that a known proposition $\varphi$ cannot be intuitionistically false (there is no proof showing that $\varphi$ is false).
We consider that combined system (the rule of intuitionistic Axiom Necessitation AN now applies to all axioms (INT), (A1)--(A8) except of (TND)) and ask the following question: Is there any reasonable replacement of the constructive concept of \textit{proof} by a classical notion of \textit{evidence} or \textit{certainty} such that the above (equivalent) principles (1)--(3) remain preserved? Such a question is natural and interesting since there are many examples where constructive or intuitionistic formalizations are extended by stronger principles which then result in a classical specification (the most prominent and probably most basic example is the extension of intuitionistic propositional logic to classical propositional logic). We propose here the following solution to our question. \\

We shall interpret the formula $\square \varphi$ as ``$\varphi$ is evident" in some intuitive, epistemic sense such that particularly the laws of classical logic are considered as evident. So the original interpretation of $\square\varphi$ as ``$\varphi$ has an actual proof" is replaced by ``$\varphi$ is evident" or ``there is evidence (certainty) of $\varphi$". Principle (1) above then reads as

\begin{equation}\label{30}
\text{``If $\varphi$ is not evident, then evidence of $\varphi$ is impossible.", i.e. }\neg\square\varphi\rightarrow\neg\Diamond\square\varphi.
\end{equation}

Note that the expression ``a proof of $\square\varphi$ is impossible" given in (1) should be translated into the classical context as ``$\square\varphi$ is impossible", i.e. ``evidence of $\varphi$ is impossible".\footnote{Of course, even if evidence of $\varphi$ is impossible, (classical) truth of $\varphi$ may be possible, i.e. $\Diamond\varphi$ may hold.}

In this way, we are able to justify and preserve modal principle (A5) in this classical context where $\square\varphi$ is interpreted as ``$\varphi$ is evident". This provides a rather static character of the concept of \textit{evidence}: If a proposition $\varphi$ is evident, then its evidence is immediately given and accessible (`here and now'); otherwise, evidence of $\varphi$ can never be established (in no accessible world of the given model), i.e. evidence is impossible. This rather strong principle, inherent in classical modal logic $S5$, is obtained by the preservation of some of the $L5$-principles (which are in accordance with BHK interpretation) for the reasoning about \textit{proof} discussed above. Also the axioms (A2)--(A4) are plausible under the new interpretation of $\square\varphi$ (since evidence of $\varphi$ is an immediately accessible fact, we may assume that it is itself evident; this can be seen as a justification for axiom (A4)). The only axiom that we have to reject is the axiom of disjunction (A1). In fact, we expect that all  classical tautologies are evident, in particular $x\vee\neg x$. But evidence of $x\vee\neg x$ cannot ensure the evidence of neither $x$ nor $\neg x$.

\section{Language and calculus}

As a name for our logical system, we propose the expression $\mathit{S5BKE}$ (`modal system $S5$ combined with belief, knowledge and evidence'). The set $Fm$ of formulas of $\mathit{S5BKE}$ is in the usual way inductively defined over a set of propositional variables $V=\{x_0, x_1,...\}$, logical connectives $\neg,\rightarrow$ and operators $\square, K, B$ for evidence, knowledge and belief, respectively. The axioms of  $\mathit{S5BKE}$ are given by the following schemes:

\begin{enumerate}
\item every formula having the form of a classical tautology is an axiom
\item $K\varphi\rightarrow\varphi$ (factivity)
\item $K\varphi \rightarrow B\varphi$ (knowledge implies belief)
\item $\square\varphi\rightarrow\square\square\varphi$
\item $\neg\square\varphi\rightarrow\square\neg\square\varphi$
\item $\square\varphi\rightarrow K\varphi$ (evidence yields knowledge, belief and truth) 
\item $K(\varphi\rightarrow\psi)\rightarrow (K\varphi\rightarrow K\psi)$ (distribution of knowledge)
\item $B(\varphi\rightarrow\psi)\rightarrow (B\varphi\rightarrow B\psi)$ (distribution of belief)
\item $\square(\varphi\rightarrow\psi)\rightarrow (\square\varphi\rightarrow\square\psi)$ (distribution of evidence)
\item $\neg B\bot$ (consistency of belief)
\end{enumerate}

Of course, this is only a minimal system of basic epistemic principles. Stronger axioms can be added (such as positive and negative introspection as well as more axioms describing the interplay between belief and knowledge, etc.) and be modeled by our semantic frameworks presented below.
Consistency of belief (axiom (x)) is generally assumed in approaches found in the literature. However, we can drop it without any problems for the resulting semantics: only minor modifications would be necessary.\\

The rules of inference are modus ponens MP and axiom necessitation AN: ``If $\varphi$ is an axiom, then $\square\varphi$ is a theorem". The notion of derivation is defined in the usual way:

\begin{definition}\label{100}
A derivation (proof, deduction) of a formula $\varphi$ from a set $\varPhi$ of formulas is a finite sequence of formulas $\varphi_1,...,\varphi_n = \varphi$ such that for every $i = 1,...,n$, $\varphi_i\in \varPhi$ or $\varphi_i$ is an axiom or $\varphi_i$ is the result of an application of modus ponens to formulas $\varphi_j$ e $\varphi_k=\varphi_j\rightarrow\varphi_i$, onde $1\le j,k < i$, or $\varphi_i$ is the result of an application of the rule of axiom necessitation to $\varphi_j$, $1\le j <i$, (in this case, $\varphi_j$ is an axiom). The number $n$ is the length of the derivation. We say that $\varphi$ is defivable from $\varPhi$ if there is a derivation of $\varphi$ from $\varPhi$, notation: $\varPhi\vdash\varphi$. A theorem is a formula derivable from the empty set. Instead of $\varnothing\vdash\varphi$ we write $\vdash\varphi$.
\end{definition}

\section{Semantics}

For Boolean algebras, we use the following notation: $\mathcal{B}=(B,\vee,\wedge,\neg,0,1)$, where $B$ is the universe and $\vee,\wedge,\neg,0,1$ are the usual operations for join, meet, complement, smallest and greatest element, respectively. Additional Boolean functions are defined in the obvious way, e.g. $a\rightarrow b :=\neg a \vee b$. Note that we do not distinguish between the symbols (logical connectives) of our object language $\neg,\rightarrow$ and corresponding symbols for the functions of an algebraic structure. The difference is generally clear by the context and we hope there is no risk of confusion. If we wish to emphasize the given algebraic structure $\mathcal{B}$, we may write $\vee^\mathcal{B}$, $\wedge^\mathcal{B}$, ... for the algebraic operations. 

\begin{definition}\label{500}
A model $\mathcal{M}$ is given by a Boolean algebra with a designated ultrafilter $\mathit{TRUE}$ and operators $\square$, $B$, $K$ for evidence, belief and knowledge, respectively, i.e. a structure $\mathcal{M}=(M,\mathit{TRUE},\vee,\wedge,\neg,0,1,\square,B,K)$, such that the following conditions are satisfied:
\begin{enumerate}
\item For any $a\in M$,
\begin{equation*}
\begin{split}
\square a =
\begin{cases}
&1,\text{ if }a=1\\
&0,\text{ else}.
\end{cases}
\end{split}
\end{equation*}
\item For every ultrafilter $U$ of the Boolean algebra, the sets 
\begin{itemize}
\item $\mathit{KNOW}(U):=\{a\in M\mid Ka\in U\}$ (``knowledge at $U$") 
\item $\mathit{BEL}(U):=\{a\in M\mid Ba\in U\}$ (``belief at $U$") 
\end{itemize}
are proper filters such that $\mathit{KNOW}(U)\subseteq U\cap \mathit{BEL}(U)$.\footnote{Note that the condition $\mathit{KNOW}(U)\subseteq U$ implies in particular that the filter $\mathit{KNOW}(U)$ is proper. If we wish to drop the consistency axiom of belief, $\neg B_\bot$, then we should admit here that $\mathit{BEL}(U)$ is not proper.}
\end{enumerate}
We refer to the elements $m\in M$ of the Boolean algebra as propositions. The designated ultrafilter $\mathit{TRUE}\subseteq M$ represents the set of true propositions. The sets of known and believed propositions are given by $\mathit{KNOW}:=\mathit{KNOW}(\mathit{TRUE})$ and $\mathit{BEL}:=\mathit{BEL}(\mathit{TRUE})$, respectively. The top element $1\in M$ stands for the evident proposition.\\
An assignment (or valuation) of a given model $\mathcal{M}$ is a function $\gamma\colon V\rightarrow M$ that extends in the canonical way to a function $\gamma\colon Fm\rightarrow M$, more precisely: $\gamma(\neg\varphi)=\neg \gamma(\varphi)$ and $\gamma(\varphi\rightarrow\psi)=\gamma(\varphi)\rightarrow\gamma(\psi)$ ($=\neg\gamma(\varphi)\vee\gamma(\psi)$). An interpretation is a tuple $(\mathcal{M},\gamma)$ formed by a model $\mathcal{M}$ and an assignment $\gamma\in M^V$. The relation of satisfaction is defined as follows:
\begin{equation*}
(\mathcal{M},\gamma)\vDash\varphi :\Leftrightarrow\gamma(\varphi)\in \mathit{TRUE}.
\end{equation*}
Of course, the definition extends to sets of formulas as follows: $(\mathcal{M},\gamma)\vDash\varPhi :\Leftrightarrow (\mathcal{M},\gamma)\vDash\varphi$ for all $\varphi\in\varPhi$.
\end{definition}

\begin{definition}\label{550}
The relation of logical consequence is defined in the usual way:
\begin{equation*}
\varPhi\Vdash\varphi :\Leftrightarrow \mathit{Mod}(\varPhi)\subseteq \mathit{Mod}(\{\varphi\})
\end{equation*}
where for any set of formulas $\varPsi$, $\mathit{Mod}(\varPsi)$ denotes the class of all intepretations satisfying $\varPsi$.
\end{definition}

The following useful observations follow from well-known facts about Boolean algebras (see, e.g. \cite{chazak}).

\begin{fact}\label{575}
Let $\mathcal{B}$ be a Boolean algebra with universe $M$. Then the following hold.
\begin{enumerate}
    \item Any proper filter is the intersection of all ultrafilters containing it.
    \secret{\item Let $F$ be a filter and $a,b \in M$. Then $a \rightarrow b \in F$ iff for any ultrafilter $U$ containing $F$, $a \in U$ implies $b \in U$.}
    \item For all $a,b \in M$, $a \leq b$ iff for any ultrafilter $U$, $a \in U$ implies $b \in U$.
\end{enumerate}
\end{fact}

\begin{theorem}[Soundness]\label{600}
For any set of formulas $\varPhi\cup\{\varphi\}$, we have
\begin{equation*}
\varPhi\vdash\varphi \Rightarrow \varPhi\Vdash\varphi.
\end{equation*}
\end{theorem}

\begin{proof}
In order to prove that the calculus is sound w.r.t. our algebraic semantics it is enough to show that all axioms and rules of inference are sound. For the rule of modus ponens, we can argue that since the set $TRUE$ is a filter, for every $a,b \in M$, if $a \in TRUE$ and $a \rightarrow b \in TRUE$, then $b \in TRUE$. In view of the rule of axiom necessitation, we have particularly to show that all axioms are interpreted by the top element of any given Boolean algebra and under any valuation. This is clear for the axioms having the form of a theorem of $\mathit{CPC}$, they are always interpreted by the top element of any Boolean algebra. As for the other axioms, let $\mathcal{M}$ be a model  with top element $1$ and a valuation $\gamma\colon V\rightarrow M$.

Consider the axiom $K\varphi\rightarrow \varphi$. Then we have $\gamma(K\varphi\rightarrow\varphi)= K\gamma(\varphi)\rightarrow\gamma(\varphi)=1$ iff $K\gamma(\varphi)\le\gamma(\varphi)$. By the definition of a model: $K\gamma(\varphi)\in U$ implies $\gamma(\varphi)\in \mathit{KNOW}(U)$ implies $\gamma(\varphi)\in U$, for any ultrafilter $U$. By Lemma \ref{575}(iii), it follows that $K\gamma(\varphi)\le\gamma(\varphi)$. For $K\varphi\rightarrow B\varphi$, we can apply the same argument, since $\mathit{KNOW}(U) \subseteq \mathit{BEL}(U)$ and then $B\gamma(\varphi) \in U$.

For $K(\varphi\rightarrow\psi) \rightarrow (K\varphi\rightarrow K\psi)$ we give a similar argument. Let $U$ be an ultrafilter and supose $\gamma (K(\varphi\rightarrow\psi)) \in U$. If $\gamma(K\varphi) \in U$, note that both $\gamma(\varphi)$ and $\gamma (\varphi\rightarrow\psi)$ are in $\mathit{KNOW}(U)$. Therefore this also holds for $\gamma(\psi)$, since $\mathit{KNOW}(U)$ is a filter. This means $\gamma(K\psi) \in U$ and we can say $\gamma(K\varphi\rightarrow K\psi) \in U$. Finally, we showed that $\gamma(K(\varphi\rightarrow\psi)) \in U$ implies $\gamma(K\varphi\rightarrow K\psi) \in U$ for any ultrafilter $U$ and therefore $\gamma(K(\varphi\rightarrow\psi)) \leq \gamma(K\varphi\rightarrow K\psi)$. The proof for $B(\varphi\rightarrow\psi) \rightarrow (B\varphi\rightarrow B\psi)$ is analogous.

As for $\square(\varphi\rightarrow\psi) \rightarrow (\square\varphi\rightarrow\square\psi)$, if the precedent is evaluated to $0$, the proof is clear. Assuming the opposite, i.e. $\gamma(\square(\varphi\rightarrow\psi)) = 1$, we have to show $\gamma(\square\varphi\rightarrow\square\psi) = 1$. From the definition of $\square$, we have $\gamma(\varphi\rightarrow\psi) = 1$ and therefore $\gamma(\varphi)\leq\gamma(\psi)$. As before, assume $\gamma(\square\varphi) = 1$. Then $\gamma(\varphi) = 1$ and also $\gamma(\psi) = 1$, leading to $\gamma(\square\varphi\rightarrow\square\psi) = 1$.

Finally, for $\square\varphi\rightarrow K\varphi$, assuming $\gamma(\square\varphi) = 1$ we have $\gamma(\varphi) = 1$. By Lemma \ref{575}(i), the set $\{1\}$ is the intersection of all ultrafilters. Then if $\gamma(K\varphi) \neq 1$, there is an ultrafilter $U$ such that $\gamma(K\varphi) \notin U$. In other words, $1 = \gamma(\varphi)\notin \mathit{KNOW}(U)$, which is impossible since $\mathit{KNOW}(U)$ is a filter. Therefore it must hold that $\gamma(K\varphi) = 1$.

The soundness of axioms $\square\varphi \rightarrow \square\square\varphi$ and $\neg\square\varphi \rightarrow \square\neg\square\varphi$ are clear by the definition of $\square$ in a model.

Soundness of the calculus now follows by induction on the length of derivations.
\end{proof}

We define the following identity connective $\equiv$ on $Fm$:\footnote{Recall that in modal logic $S5$, the formulas $\square(\varphi\rightarrow\psi)\wedge \square(\psi\rightarrow\varphi)$ and $\square((\varphi\rightarrow\psi)\wedge (\psi\rightarrow\varphi))$ are logically equivalent as $\square$ distributes over conjunction.}
\begin{equation*}
\varphi\equiv\psi :=\square(\varphi\leftrightarrow\psi).
\end{equation*}

Intuitively, $\varphi\equiv\psi$ says that $\varphi$ and $\psi$ denote the same proposition. We will call $\equiv$ the connective of propositional identity. In fact, we have the following result:

\begin{lemma}\label{620}
For any model $\mathcal{M}$ and any assignment $\gamma\colon V\rightarrow M$,
\begin{equation*}
(\mathcal{M},\gamma)\vDash\varphi\equiv\psi \Leftrightarrow \gamma(\varphi)=\gamma(\psi).
\end{equation*}
\end{lemma}

Also on the syntactical level, the connective $\equiv$ behaves as expected: 

\begin{theorem}\label{640}
The Substitution Principle holds, i.e.\ the following scheme is a theorem:\footnote{$\varphi[x:=\psi]$ is the result of substituting $\psi$ for every ocorrence of variable $x$ in $\varphi$.}
\begin{equation*}
    (\varphi\equiv\psi)\rightarrow (\chi[x:=\varphi]\equiv\chi[x:=\psi]).
\end{equation*}
\end{theorem}

\begin{theorem}\label{650}
$\square\varphi\leftrightarrow (\varphi\equiv\top)$ is a theorem. (Actually, $\square\varphi\equiv (\varphi\equiv\top)$ is a theorem.)
\end{theorem}

\begin{proof}
$\varphi\leftrightarrow (\varphi\leftrightarrow\top)$ is a theorem of $\mathit{CPC}$. Now, apply axiom necessitation and distribution.
\end{proof}

\begin{theorem}[Completeness]\label{660}
For any set $\varPhi\cup\{\varphi\}\subseteq Fm$, we have
\begin{equation*}
\varPhi\Vdash\varphi\Rightarrow\varPhi\vdash\varphi.
\end{equation*}
\end{theorem}

\begin{proof}
Suppose $\varPhi\nvdash\varphi$. Exactly as in $\mathit{CPC}$, it follows that $\varPhi\cup\{\neg\varphi\}$ is consistent. If there is an interpretation satisfying that set, then $\varPhi\Vdash\varphi$ is impossible. So it is enough to show that any consistent set is satisfiable. Let $\varPsi$ be a consistent set of formulas. Zorn's Lemma implies the existence of a maximal consistent extension $\varGamma\supseteq\varPsi$. We define a relation on $Fm$ by 
\begin{equation*}
\varphi\approx\psi :\Leftrightarrow \varGamma\vdash\varphi\equiv\psi.
\end{equation*}

Then $\approx$ is a congruence relation, i.e.\ it is an equivalence relation which is compatible with the logical connectives and modal and epistemic operators of the language. We write $\overline{\varphi}$ for the congruence class of $\varphi$ modulo $\approx$, and $\overline{\varPhi}:=\{\overline{\varphi}\mid\varphi\in\varPhi\}$. In the following, we construct a model $\mathcal{M}$ with universe $M:=\overline{Fm}$, $0:=\overline{\bot}$, $1:=\overline{\top}$ and designated ultrafilter $\mathit{TRUE}:=\overline{\varGamma}$. Since $\approx$ is a congruence relations, the operations $\overline{\varphi}\vee\overline{\psi} :=\overline{\varphi\vee\psi}$, $\overline{\varphi}\wedge\overline{\psi} :=\overline{\varphi\wedge\psi}$, $\neg\overline{\varphi}:=\overline{\neg\varphi}$, $\square\overline{\varphi}:=\overline{\square\varphi}$,
$B\overline{\varphi}:=\overline{B\varphi}$,
$K\overline{\varphi}:=\overline{K\varphi}$
are  well-defined. Also, if $\varphi\in\varGamma$ and $\varphi\approx\psi$, then $\psi\in\varGamma$. We show that
\begin{equation*}
\mathcal{M}=(M,\mathit{TRUE},\vee,\wedge,\neg,0,1,\square,B,K)
\end{equation*}
is a model. $\mathcal{M}$ is a Boolean algebra because all equations that define the class of Boolean algebras are satisfied. In fact, for every classical tautology of the form $\varphi\leftrightarrow\psi$, $\varphi\equiv\psi$ is a theorem (apply the rule of axiom necessitation) and therefore we get $\varphi\approx\psi$. In particular, $\mathit{TRUE}$ is an ultrafilter of the Boolean algebra. It remains to check the conditions (i) and (ii) of Definition \ref{500}.\\
Condition (i) holds: Recall that by Theorem \ref{650}, $\square\varphi\leftrightarrow (\varphi\equiv\top)$ is a theorem and therefore an element of $\varGamma$. So $\overline{\varphi}=\overline{\top}$ iff $\varphi\approx\top$ iff $\varGamma\vdash\varphi\equiv\top$ iff $\square\varphi\in\varGamma$ iff $\overline{\square\varphi}\in\overline{\varGamma}$ iff  $\square\overline{\varphi}\in\overline{\varGamma}$. Thus,
\begin{equation*}
    \square\overline{\varphi}\in\overline{\varGamma}\Leftrightarrow \overline{\varphi}=\overline{\top},
\end{equation*}
for all $\varphi\in Fm$. Now, it is enough to show that the following hold:
\begin{equation*}
\begin{split}
&\square\overline{\varphi}\in\overline{\varGamma} \Rightarrow \square\overline{\varphi}=\overline{\top}\\
&\square\overline{\varphi}\notin\overline{\varGamma} \Rightarrow \square\overline{\varphi}=\overline{\bot}.
\end{split}
\end{equation*}
By axiom (iv) and the above equivalence, we get: $\square\overline{\varphi}\in\overline{\varGamma}$ implies $\square\varphi\in\varGamma$ implies $\square\square\varphi\in\varGamma$ implies $\square\overline{\square\varphi}\in\overline{\varGamma}$ implies $\overline{\square\varphi}=\overline{\top}$ implies $\square\overline{\varphi}=\overline{\top}$.\\
By the maximal consisistency of $\varGamma$, axiom (v) and the above equivalence, we get:
$\square\overline{\varphi}\notin\overline{\varGamma}$ implies $\neg\square\overline{\varphi}\in\overline{\varGamma}$ implies $\square\overline{\neg\square\varphi}\in\overline{\varGamma}$ implies $\neg\overline{\square\varphi}=\overline{\top}$ implies $\overline{\square\varphi}=\overline{\bot}$ implies $\square\overline{\varphi}=\overline{\bot}$.\\
Then condition (i) follows:
\begin{equation*}
\begin{split}
\square\overline{\varphi} =
\begin{cases}
&\overline{\top},\text{ if }\overline{\varphi}=\overline{\top}\\
&\overline{\bot},\text{ else}.
\end{cases}
\end{split}
\end{equation*}
Condition (ii) follows immediately from the axioms. Recall that $\mathit{KNOW}(U)$ is a filter iff the following holds: $\overline{\top}\in\mathit{KNOW}(U)$ and $\overline{\varphi}\rightarrow\overline{\psi},\overline{\varphi}\in \mathit{KNOW}(U)$ implies $\overline{\psi}\in\mathit{KNOW}(U)$.
\end{proof} 

\begin{definition}\label{1000}
A frame $\mathcal{F}=(W,P,E_K,E_B,w_T)$ is given by
\begin{itemize}
    \item a non-empty set $W$ of worlds with a designated world $w_T\in W$,
    \item a set $P\subseteq Pow(W)$ which contains $\varnothing$ and $W$ and is closed under the following conditions: if $A,B\in P$, then the sets $A\cap B$, $A\cup B$, $\sim A := P\smallsetminus A$, $KA := \{w\in W\mid A\in E_K(w)\}$ and $BA := \{w\in W\mid A\in E_B(w)\}$ are elements of $P$,
    \item a function $E_K\colon W\rightarrow Pow(P)$ such that for each $w\in W$, $E_K(w)\subseteq P$ is a filter on $P$, i.e. $E_K(w)$ is a non-empty set satisfying the following conditions: if $A,B\in E_K(w)$, then $A\cap B\in E_K(w)$; and if $A\in E_K(w)$ and $A\subseteq B\in P$, then $B\in E_K(w)$; and furthermore: for each $w\in W$, $E_K(w)$ satisfies the following factivity condition: $A\in E_K(w)$ implies $w\in A$. 
    \item a function $E_B\colon W\rightarrow Pow(P)$ such that for each $w\in W$, $E_B$ is a proper filter on $P$ and $E_K(w)\subseteq E_B(w)$.
\end{itemize}
\end{definition}

Intuitively, $P$ is the set of all propositions, $E_K(w)$ is the set of propositions known at world $w$, $E_B(w)$ is the set of propositions believed at world $w$, $KA$ is the proposition saying ``$A$ is known" and $BA$ is the proposition saying ``$A$ is believed".

\begin{definition}\label{1100}
Given a frame $\mathcal{F}=(W,P,E_K,E_B)$, an assignment (or valuation) is a function $g\colon V\rightarrow P$. The tuple $\mathcal{K}=(\mathcal{F},g)$ is called a frame-based model, and its relation of satisfaction between worlds $w\in W$ and formulas is defined as follows:
\begin{itemize}
    \item $w\vDash x :\Leftrightarrow w\in g(x)$
    \item $w\vDash \neg\varphi :\Leftrightarrow w\nvDash\varphi$
    \item $w\vDash\varphi\rightarrow\psi :\Leftrightarrow w\nvDash\varphi$ or $w\vDash\psi$
    \item $w\vDash\square\varphi :\Leftrightarrow w'\vDash\varphi$, for all $w'\in W$
    \item $w\vDash K\varphi :\Leftrightarrow \varphi^*\in E_K(w)$, where $\varphi^* :=\{w\in W\mid w\vDash\varphi\}$
    \item $w\vDash B\varphi :\Leftrightarrow \varphi^*\in E_B(w)$
\end{itemize}
Finally, for any frame-based model $\mathcal{K}=(\mathcal{F},g)$ and any formula $\varphi\in Fm$, we define:
\begin{equation*}
    \mathcal{K}\vDash\varphi :\Leftrightarrow w_T\vDash\varphi.
\end{equation*}
The frame-based relation of logical consequence then is given in the obvious way: $\varPhi\Vdash_f\varphi :\Leftrightarrow \mathcal{K}\vDash\varPhi$ implies $\mathcal{K}\vDash\varphi$ for any frame-based model $\mathcal{K}$.
\end{definition}

\begin{remark}\label{1180}
Let $\mathcal{K}=(\mathcal{F},g)$ be a frame-based model with a set $P$ of propositions. We extend $g\colon V\rightarrow P$ to a function $g\colon Fm\rightarrow P$ defining recursively:
\begin{itemize}
\item $g(\neg\varphi) :=P\smallsetminus g(\varphi)$
\item $g(\varphi\rightarrow\psi) := (P\smallsetminus g(\varphi))\cup g(\psi)$
\item $g(K\varphi) := \{w\in W\mid g(\varphi)\in E_K(w)\}$ \item $g(B\varphi) := \{w\in W\mid g(\varphi)\in E_B(w)\}$ \item $g(\square\varphi)=W$ if $g(\varphi)=W$; $g(\square\varphi)=\varnothing$ otherwise.
\end{itemize}
Then by induction on formulas simultaneous for all worlds $w\in W$, one shows the following for any $\varphi\in Fm$ and any $w\in W$: $w\vDash\varphi\Leftrightarrow w\in g(\varphi)$.\\ 
Thus, $g(\varphi)=\{w\in W\mid w\vDash\varphi\}=\varphi^*\in P$. 
\end{remark}

\begin{remark}[The collapsing problem of introspective knowledge]\label{1190}
In contrast to the usual relational semantics of modal epistemic logic found in the literature (cf. \cite{fag, mey}), we do not model knowledge (belief) by means of a specific accessibility relation. Instead, we represent knowledge by a function that assigns to each world sets of propositions. In this way, we avoid that strong forms of knowledge collapse into evidence. Suppose, for instance, knowledge is represented by the usual accessibility relation and satisfies positive and negative introspection, i.e. the $S5$-style axioms. Then knowledge is given by the universal relation exactly as evidence. That is, knowledge collapses into evidence. This, however, is not derivable from the axioms. Our modeling of knowledge avoids the collapsing problem on the semantic level.  
\end{remark}

\begin{theorem}\label{1200}
For every (algebraic) model $\mathcal{M}$ and every assignment $\gamma\colon V\rightarrow M$ there is a frame-based model $\mathcal{K}=(\mathcal{F},g)$ such that for all formulas $\varphi\in Fm$,
\begin{equation*}
    (\mathcal{M},\gamma)\vDash\varphi\Leftrightarrow \mathcal{K}\vDash\varphi.
\end{equation*}
\end{theorem}

\begin{proof}
Suppose $\mathcal{M}$ is a model and $\gamma\in M^V$ is an assignment. Let $W$ be the set of all ultrafilters of the Boolean algebra underlying $\mathcal{M}$, and let $w_T :=\mathit{TRUE}\in W$ be the designated world. For each $m\in M$, we define 
\begin{equation*}
m^+ :=\{w\in W\mid m\in w\}.
\end{equation*}
Then 
\begin{equation*}
P:=\{m^+\mid m\in M\}
\end{equation*}
is the set of propositions of the desired frame, and
\begin{equation*}
\begin{split}
&E_K(w) := \{m^+\subseteq W\mid m\in KNOW(w)\},\\ 
&E_B(w) := \{m^+\subseteq W\mid m\in BEL(w)\}
\end{split}
\end{equation*}
are the sets of known and believed propositions at $w$, respectively. In fact, for all $m,m'\in M$ we have
\begin{equation*}
    \begin{split}
        &m^+\cap m'^+ = (m\wedge m')^+\\
        &m^+\cup m'^+ = (m\vee m')^+\\
        &\sim (m^+) = (\neg m)^+\\
        &K(m^+) = \{w\in W\mid m^+\in E_K(w)\} = (Km)^+\\
        &B(m^+) = \{w\in W\mid m^+\in E_B(w)\} = (Bm)^+\\
        &W = 1^+\\
        &\varnothing = 0^+
    \end{split}
\end{equation*}
Thus, $P$ satisfies the required closure properties. Then one easily recognizes that also the sets $E_K(w)$ and $E_B(w)$ are closed under intersection. Suppose $m^+\in E_K(w)$ and $m^+\subseteq m'^+$ for some $m'\in M$. Then for any ultrafilter $w'$ of the Boolean algebra underlying model $\mathcal{M}$ we have: $m\in w'$ implies $m'\in w'$. This is equivalent to the condition that $m\rightarrow m'$ belongs to all ultrafilters, i.e. $m\rightarrow m' = 1$. In particular, $m\rightarrow m'\in KNOW(w)$ and therefore $(m\rightarrow m')^+\in E_K(w)$. Since $m^+\in E_K(w)$, we get $m'^+\in E_K(w)$. Thus, $E_K(w)$ is a filter. Similarly, one shows that $E_B(w)$ is a filter. Finally, we show the factivity condition for $E_K(w)$. Suppose $m^+\in E_K(w)$. Then $m\in KNOW(w)\subseteq w$, by definition of $KNOW(w)$. Thus, $w\in m^+$. We have shown that $\mathcal{F}=(W,E_K,E_B,w_T)$ is a frame. The assignment $g\colon V\rightarrow P$ is defined by $g(x) :=\gamma(x)^+$. Then by induction and Remark \ref{1180}, 
\begin{equation*}
    g(\varphi) = \gamma(\varphi)^+,
\end{equation*}
for all $\varphi\in Fm$. Then for any $w\in W$ and $\varphi\in Fm$:
\begin{equation*}
    w\vDash\varphi\Leftrightarrow w\in\varphi^*\Leftrightarrow w\in g(\varphi)\Leftrightarrow w\in \gamma(\varphi)^+\Leftrightarrow \gamma(\varphi)\in w.
\end{equation*}
In particular:
\begin{equation*}
    (\mathcal{F},g)\vDash\varphi\Leftrightarrow w_T\vDash\varphi\Leftrightarrow \gamma(\varphi)\in w_T=\mathit{TRUE}\Leftrightarrow (\mathcal{M},\gamma)\vDash\varphi.
\end{equation*}
\end{proof}

\begin{corollary}\label{1300}
Our logic is complete w.r.t. frame-based semantics: $\varPhi\Vdash_f\varphi$ implies $\varPhi\vdash\varphi$.
\end{corollary}

\begin{proof}
Suppose $\varPhi\nvdash\varphi$. Then by completeness w.r.t. algebraic semantics, $\varPhi\nVdash\varphi$. Thus, there is a model $(\mathcal{M},\gamma)$ such that $(\mathcal{M},\gamma)\vDash\varPhi$ and $(\mathcal{M},\gamma)\nvDash\varphi$. Then Theorem \ref{1200} yields a frame-based model with that property. Hence, $\varPhi\nVdash_f\varphi$.
\end{proof}

\end{document}